\documentclass[sigconf]{acmart}

\usepackage{algorithm}
\usepackage[noend]{algorithmic}

\newtheorem{corollary}{Corollary}

\newtheorem{definition}{Definition}

\usepackage{booktabs} 

\usepackage{enumitem}
\setlist[itemize]{leftmargin=*}
\setlist[enumerate]{leftmargin=*}
\setcopyright{acmcopyright}

\copyrightyear{2022} 
\acmYear{2022} 
\setcopyright{acmcopyright}\acmConference[SAC '22]{The 37th ACM/SIGAPP Symposium on Applied Computing}{April 25--29, 2022}{Virtual Event, }
\acmBooktitle{The 37th ACM/SIGAPP Symposium on Applied Computing (SAC '22), April 25--29, 2022, Virtual Event, }
\acmPrice{15.00}
\acmDOI{10.1145/3477314.3507064}
\acmISBN{978-1-4503-8713-2/22/04}

\begin{document}
\title{K-Core Decomposition on Super Large Graphs with \\ Limited Resources}









\author{Shicheng Gao$^{\S}$, Jie Xu${^\ddagger}$, Xiaosen Li${^\ddagger}$, Fangcheng Fu${^\dagger}$, Wentao Zhang${^\dagger}$, Wen Ouyang${^\ddagger}$, \\Yangyu Tao${^{\ddagger}}$, Bin Cui${^{\dagger\triangle\star}}$}
\affiliation{$^{\S}$ School of Electronic and Computer Engineering of Shenzhen Graduate School, Peking University~~~~~$^\ddagger$Tencent Inc.\\ $^\triangle$Center for Data Science, Peking University \& National Engineering Laboratory for Big Data Analysis and Applications\\ $^\dagger$School of CS \& Key Laboratory of High Confidence Software Technologies, Peking University\\ $^{\star}$Institute of Computational Social Science, Peking University (Qingdao), China}
\affiliation{$^{{\S}\dagger}$\{gscim.d,ccchengff,wentao.zhang,bin.cui\}@pku.edu.cn\\$^\ddagger$\{joyjxu,hansenli,gdpouyang,brucetao\}@tencent.com}

\renewcommand{\shortauthors}{Shicheng Gao, et al.}

\begin{abstract}
K-core decomposition is a commonly used metric to analyze graph structure or study the relative importance of nodes in complex graphs.
Recent years have seen rapid growth in the scale of the graph, especially in industrial settings.
For example, our industrial partner runs popular social applications with billions of users and is able to gather a rich set of user data.
As a result, applying K-core decomposition on large graphs has attracted more and more attention from academics and the industry.
A simple but effective method to deal with large graphs is to train them in the distributed settings, and some distributed K-core decomposition algorithms are also proposed.
Despite their effectiveness, we experimentally and theoretically observe that these algorithms consume too many resources and become unstable on super-large-scale graphs, especially when the given resources are limited. 
In this paper, we deal with those super-large-scale graphs and propose a divide-and-conquer strategy on top of the distributed K-core decomposition algorithm. 
We evaluate our approach on three large graphs. The experimental results show that the consumption of resources can be significantly reduced, and the calculation on large-scale graphs becomes more stable than the existing methods.
For example, the distributed K-core decomposition algorithm can scale to a large graph with 136 billion edges without losing correctness with our divide-and-conquer technique.
\end{abstract}

%
%



\maketitle

{\fontsize{8pt}{8pt} \selectfont
\textbf{ACM Reference Format:}\\
Shicheng Gao, Jie Xu, Xiaosen Li, Fangcheng Fu, Wentao Zhang, Wen Ouyang, Yangyu Tao, Bin Cui. 2022. K-Core Decomposition on Super Large Graphs with Limited Resources. In \textit{The 37th ACM/SIGAPP Symposium on Applied Computing (SAC '22), April 25–29, 2022, Virtual Event, . } ACM, New York, NY, USA, 10 pages. https://doi.org/10.1145/3477314.3507064 }

\section{Introduction}
\label{sec:intro}

\subsubsection*{\underline{Background and Motivation}}
Relationships between individuals or entities can be captured as graphs where each node (a.k.a. vertex) represents an individual or entity, and each edge represents the corresponding connection or relation. Capturing graph structure of data is useful in many applications, 
such as targeted advertising~\cite{yang2008discovering,wu2020graph}, knowledge distillation~\cite{zhang2020reliable,zhang2021rod}, data annotation~\cite{DBLP:journals/pvldb/ZhangYWSLW021,zhang2021rim,zhang2021alg}, and protein analysis~\cite{,wang2010fast,gutierrez2014biological}.
Each node in a graph often exhibits a crucial property --- the importance or utility of a node depends on the number of connections between it and other nodes. Especially in social networks, the engagement of a node is more likely to happen if its most neighbors are engaged. To this end, the concept of K-core is introduced to estimate the coreness of a node. A K-core subgraph of graph $G$ is the largest induced subgraph where every vertex has a degree of at least $k$.

The most prevailing way to estimate the coreness of a node is by the K-core decomposition, which has widespread adoption in many areas, such as network structure visualization, Internet topology exploration, and community-related problems~\cite{carmi2007model,alvarez2006large, serrano2009extracting}. Furthermore, it works as a fundamental building block for a variety of problems, such as finding the approximation for the densest subgraph problem or the densest at-least-$k$-subgraph problem~\cite{lee2010survey}. Consequently, K-core decomposition has become a popular topic in both academic and industrial communities.
Notably, driven by the explosive surge of data volume and increasing complexity of social networks, 
many efforts have been devoted to accomplishing the K-core decomposition in a distributed manner~\cite{montresor2012distributed, mandal2017distributed, pechlivanidou2014mapreduce}. 

One of our industrial partners runs popular social applications with billions of users and can extract various graphs, which are then utilized to mine the user relationships for other applications, such as user profiling and machine learning tasks. 
To enhance the model performance, our industrial partner attempts to process larger graphs, i.e., more nodes and more edges. Therefore, distributed $k$-core decomposition has become a widely-used application in its production pipeline. 

\subsubsection*{\underline{Challenges}}
Although the increment in data volume is beneficial for mining more precise user information, significant performance degradation is encountered. To summarize, we observe two severe issues when processing large graphs.

First, since the production cluster is shared by many tasks, the maximal resources allowed for each task are limited, which becomes insufficient when the scale of graphs increases. Take a real case in our industrial partner as an example, where the graph is composed of 2.2 billion nodes and 136 billion edges. It takes several terabytes of memory to store such a huge graph. 
To the best of our knowledge, the existing distributed graph processing engines keep the entire graph in memory throughout the execution. Therefore, none of these existing works can support $k$-core decomposition for such a large graph within the resource limitation. 

Second, even for graphs that can be supported with the resource constraints, the performance is undesirable when the graphs are enlarged due to the increase in ``communication amount''. During the execution, the $k$-core decomposition algorithm iteratively updates the coreness value of each node according to its neighbors. In the distributed setting, every update triggers the communication between different workers to ensure that all neighbors can retrieve up-to-date information. We call the total number of communicated updates as the communication amount. Undoubtedly, when there are tremendous edges, the number of updated nodes increases correspondingly, leading to non-trivial communication overhead and thus dampening the efficiency.

\subsubsection*{\underline{Summary of Contributions}}
As discussed above, the undesirable performance mainly comes from the large scale of graphs. To tackle these challenges, we develop a divide-and-conquer strategy, namely \emph{DC-$k$Core}, to support extremely large graphs under certain resource constraints and accelerate the total running time cost. To the best of our knowledge, this is the first work that applies the divide-and-conquer strategy for the $k$-core decomposition problem. The major contributions of this work are summarized below. 

\paragraph{The Divide Step: Graph Division}
To support extremely large graphs under limited resources, we propose to divide the graphs into different parts, where each part is a subgraph of the original graph. We devise a brand new graph division strategy that matches the coreness computation characteristics, rather than randomly distributing the original graphs in the vertex- or edge-centric manner. Specifically, since the nodes with small coreness values have zero impact on the computation for nodes with large coreness values, we propose dividing the graphs by a given degree threshold. By doing so, different subgraphs can be processed individually, reducing the peak resource requirement.  

\paragraph{The Conquer Step: Subgraph Decomposition}
After an original graph is divided into multiple subgraphs, we introduce to apply $k$-core decomposition on each subgraph individually. However, since the edges connected between different subgraphs are cut off, part of the neighbor information has been lost. To ensure the correctness of subgraph decomposition, we generate the external information, which can be reckoned as the summary of neighbor information between different subgraphs. A novel algorithm that decomposes the subgraph with the help of external information is developed. Compared with decomposing the original graphs directly, the utilization of external information can reduce the communication amount to a large extent, boosting the overall performance. 

\paragraph{Deployment and Evaluation on Real-World Applications}
We implement the proposed divide-and-conquer strategy, namely \emph{DC-$k$Core}, on top of Apache Spark and the Angel parameter server platform~\cite{zaharia2010spark,jiang2018angel}. \emph{DC-$k$Core} has been widely deployed in the production pipeline of our industrial partner. There are apps adopting \emph{DC-$k$Core} to execute \emph{KOL(key opinion leader)} analysis, or using it to detect the merchant from massive Payment networks. Besides, it is also used in fraud-detection for risk control. Comprehensive experiments are conducted to evaluate the effectiveness of \emph{DC-$k$Core}. First, empirical results show that \emph{DC-$k$Core} is more efficient than the existing distributed $k$-core decomposition implementations and supports a much larger scale of graphs with up to 136 billion edges. To the best of our knowledge, this is the first work supporting such a large graph scale. Second, we empirically verify that \emph{DC-$k$Core} is also efficient for small- or medium-scale of graphs since it can reduce the communication amount significantly by summarizing the neighbor information. Moreover, we evaluate the sensitivity of \emph{DC-$k$Core} in terms of the number of subgraphs or choices of division strategies, shedding light on the potential improvement on a wider range of applications.

\section{related work}
\label{sec:related_work}
The definition of \emph{K-core} is first brought out by Serdman\cite{seidman1983network} to characterize the cohesive regions of graphs. There is a considerable amount of research\cite{luczak1991size, janson2007simple, molloy2005cores, cooper2004cores, pittel1996sudden} done to study the existence of non-empty k-core in a random graph. 

Thanks to the well-defined structure of k-core, it plays an important role in the analysis of the structure of certain types of networks\cite{dorogovtsev2006k} and generating graphs with specific properties\cite{baur2008augmenting}. There exist many graph problems such as maximal clique finding\cite{balasundaram2011clique}, dense subgraph discovery\cite{andersen2009finding}, and betweenness approximation\cite{healy2006characterization} using k-core decomposition as a subroutine.

The first K-core decomposition algorithm was proposed by Batagelj and Zavernik (BZ)\cite{batagelj2011fast}. In their work, the coreness of each node is generated by recursively removing nodes of degree less than $k$. It requires random access to the entire graph. Therefore, in order to ensure good efficiency, the entire graph should be stored in memory while processing. As the scale of the graph increases, memory cannot support the calculation of the whole graph. As a result, K-core decomposition algorithm using the combination of memory and external memory was developed. The first K-core decomposition algorithm (EMCore) using secondary storage was proposed by Cheng et al. in \cite{cheng2011efficient}. EMCore uses a strategy based on the BZ algorithm which requires $O(k_{max})$ scans of the graph, where $k_{max}$ is the largest coreness value of the graph. \cite{khaouid2015k} optimized the implementation of EMCore and can handle networks of billions of edges using a single consumer-level machine and can produce excellent approximations in only a fraction of the execution time.

\cite{esfandiari2018parallel} computes an approximation of nodes coreness by sampling edges into subgraphs. The computation of each subgraph follows a way of peeling nodes.

With the increasing graph scale, the idea of distributed K-core decomposition was first proposed in \cite{montresor2012distributed}, where \emph{node index} was introduced. 
This work has been implemented by Mandal et al.\cite{mandal2017distributed} on Spark. 
The realization of k-core decomposition on PSGraph is also based on \cite{montresor2012distributed}. 
This algorithm was further extended to dynamic graphs by Aridhi et al.\cite{aridhi2016distributed}. Moreover, there are I/O efficient algorithms\cite{cheng2011efficient, wen2016efficient} aiming at handling large-scale graphs which cannot fit into memory proposed, based on the distributed algorithm proposed by Montresor et al.\cite{montresor2012distributed}. 
There is also another distributed algorithm using MapReduce\cite{pechlivanidou2014mapreduce, dean2004mapreduce}.

K-core decomposition are also studied on various kinds of networks, such as on directed\cite{giatsidis2013d} and weighted\cite{giatsidis2011evaluating} network. Moreover, Li et al.\cite{li2013efficient} deal with the problem of k-
core maintenance in large dynamic graphs. There are many works that deal with k-core decomposition on dynamic graphs\cite{sariyuce2016incremental, miorandi2010k, jakma2012distributed}. And the k-core decomposition in the streaming scenario was first discussed in \cite{sariyuce2013streaming}.

\section{preliminaries}
\label{sec:preliminaries}
\begin{table}[!t]
  \caption{Notations}
  \label{tab:notations}
  \begin{tabular}{ccl}
    \toprule
    Symbols&Definitions\\
    \midrule
    $\mathcal{N}(v)$ & Neighbor set of node $v$\\
    $deg(v)$ & Number of neighbors of node $v$ (a.k.a., degree) \\
    $G_k(V_k, E_k)$ & The $k$-core subgraph of $G$ (abbr. $G_k$) \\
    $G_{!k}(V_{!k}, E_{!k})$ & The rest part of $G$ by removing $G_k$ (abbr. $G_{!k}$) \\
    $G_{\hat{k}}$ & The rough $k$-core subgraph of $G$ \\
    $\mathcal{N}_{G_k}(v)$ & Neighbor set in $G_k$ of node $v$ \\
    $deg_{G_k}(v)$ & Number of neighbors in $G_k$ of node $v$ \\
    $coreness(v)$ & the coreness value of node $v$\\
    $\mathcal{E}(v)$ & Number of neighbors in $G_k$ of node $v\in G_{!k}$ \\
  \bottomrule
\end{tabular}
\end{table}

\subsection{Problem Definition}
A \textit{network} or \textit{graph} is denoted by $G(V, E)$, where $V$ is the set of vertices (nodes) and $E$ is the set of edges (links). We use the symbol $n$ to represent the number of nodes ($n = |V|$) and the symbol $m$ to represent the number of edges ($m = |E|$). Given a node $v \in V$, the set of its neighbors are denoted by $\mathcal{N}_G(v)$, i.e., the set of all nodes adjacent to $v$. We use $deg_G(v)$ to denote the degree (number of neighbors) of node $v$ ($deg_G(v) = |\mathcal{N}_G(v)|$).

Given an undirected and unweighted graph $G(V, E)$, 
the concept of K-core decomposition is considered as the following definitions. 

\begin{definition}[K-core Subgraph]
\label{def:kcore}
Let $H(C, E|C)$  be a subgraph of $G(V, E)$, where $C \subseteq V, E|C = \{(u,v)\in E: u \in C \land v \in C\}$.
$H$ is defined to be the K-core subgraph of $G$, denoted by $G_k$, if it is the maximal subgraph of $G$ satisfying the following property: $\forall v \in C:deg_H(v)\geq k$.
\end{definition}


\begin{definition}[Coreness Value]
\label{def:coreness}
A node $v \in V$ has coreness value $coreness(v) = k$, if it belongs to the k-core subgraph but not to the (k+1)-core subgraph.
\end{definition}

\begin{figure}[!t]
  \centering
  \setlength{\abovecaptionskip}{-0.1cm}   
  \setlength{\belowcaptionskip}{-0.2cm}   
  \includegraphics[width=0.8\linewidth]{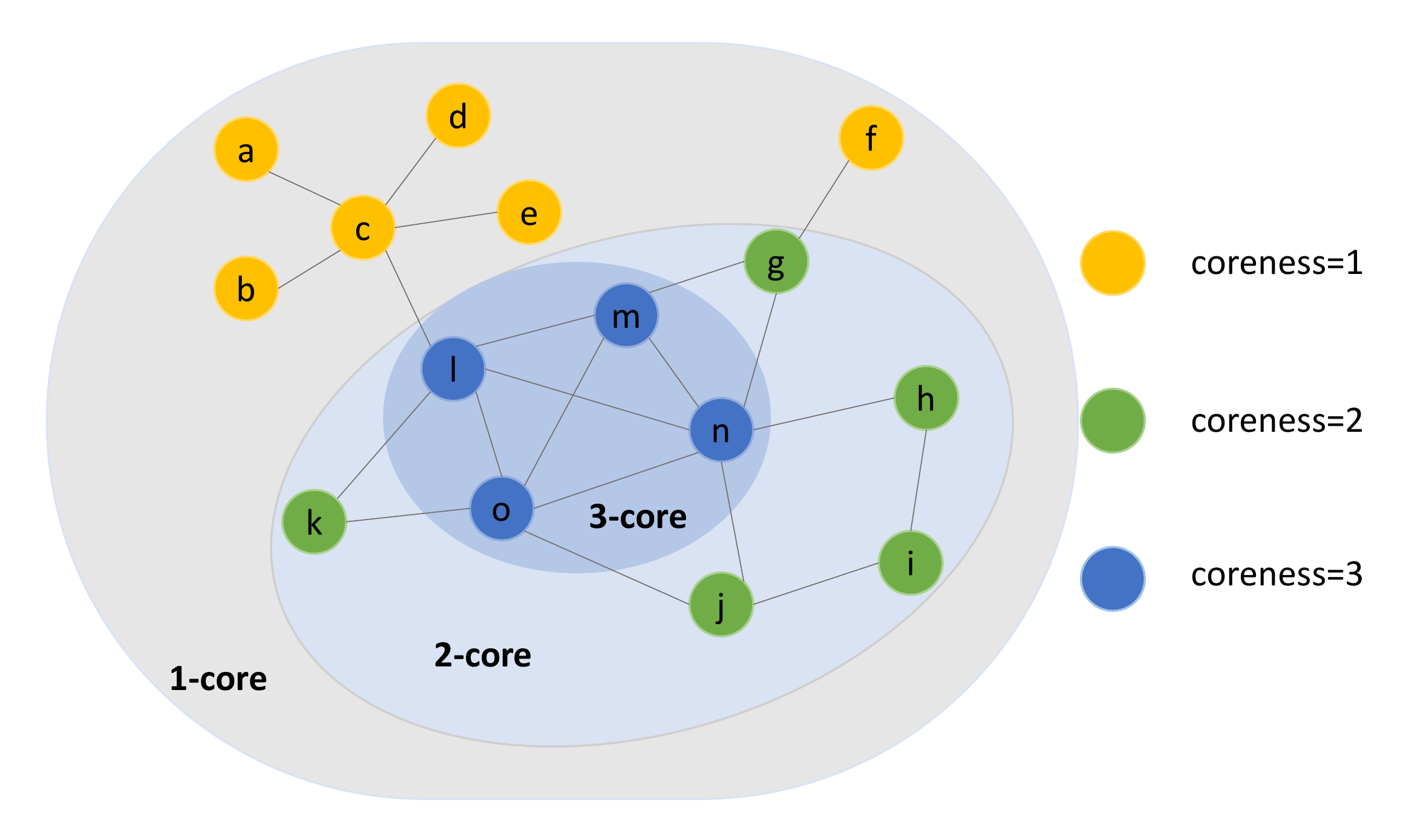}
  \caption{A instance graph with max coreness 3.}
  \label{fig:samplegraph}
\end{figure}

For the example shown in Figure~\ref{fig:samplegraph}, the graph $G$ has 15 nodes, and each k-core subgraph is surrounded by a circle.
Given an undirected and unweighted graph $G(V, E)$, the goal of \textbf{k-core decomposition} is to \emph{find all the k-core subgraphs $\{G_k\}_{k=0}^{k_{max}}$, where $k_{max} = max\{coreness(v) | \forall v \in V\}$}.
Obviously, this problem is equivalently to \emph{find the coreness values for all nodes in $G$.}

\begin{algorithm}[!t] 
\caption{Procedure to estimate the coreness value for a node $v$ in one iteration.} 
\label{alg:hindexestimate} 
\begin{algorithmic}[1]
\REQUIRE ~~\\
$\mathcal{N}(v)$: Neighbors of node $v$;
\ENSURE ~~\\
Estimation for the coreness value of node $v$;

\STATE $Cores \leftarrow$ Estimation of coreness values of all nodes in $\mathcal{N}(v)$ from the previous iteration;
\STATE $Cores \leftarrow Cores.sorted.reverse$
\STATE $i \leftarrow 0$
\STATE $C_v \leftarrow Cores.length$

\WHILE{$i < Cores.length$}
    \IF{$Cores(i) \geq i + 1$}
        \STATE $i \leftarrow i + 1$;
    \ELSE
        \STATE $C_v \leftarrow i$;
        \STATE break while;
    \ENDIF
\ENDWHILE
\RETURN $C_v$;
\end{algorithmic}
\end{algorithm}

\subsection{Distributed K-core Decomposition}
Owing to the incredible surge of data volume, it is impossible to process the entire graph on a single machine. Thus, several works were developed to accomplish the k-core decomposition in a distributed manner. 
The first distributed k-core decomposition algorithm was reported by~\citet{montresor2012distributed}.
The essential idea is to update the coreness values by calculating the \textbf{node index} iteratively. Due to space constraint, we refer interested readers to~\cite{lu2016h} for more details on how to estimate the coreness values via the node index. 
As shown in Algorithm~\ref{alg:hindexestimate}, in each iteration, each node receives the estimated coreness values of its neighbors from the previous iteration, estimates its own coreness value, and sends the new estimation to its neighbors for the next iteration. 
The estimation will converge to the exact coreness value, which is proven by~\citet{montresor2012distributed}, 
and the decomposition accomplishes when the estimated coreness values for all nodes become stable.

\citet{mandal2017distributed} implement the distributed $k$-core decomposition algorithm on top of the GraphX platform, which is a popular distributed graph processing engine based on Apache Spark. However, since the resilient distributed datasets (RDDs) of Spark are immutable, the implementation in GraphX needs to reconstruct the RDDs frequently and suffers from heavy communication cost. So, the efficiency is insufficient for many industrial-scale graphs. 

PSGraph~\cite{jiang2020psgraph} is a graph processing platform that supports exetremly large-scale graph. PSGraph is also built upon Spark in order to be production-friendly. However, unlike GraphX, PSGraph leverages the power of parameter server to update the estimated coreness values so that it does not need to reconstruct the graph data repeatedly. Therefore, PSGraph is much more efficient than GraphX and supports much larger scale of graphs with up to billions of edges. 
Nevertheless, since it stores the entire graph during the $k$-core decomposition task, extremely huge graphs cannot be supported under the resource limitation of our productive environment. 

\section{Our Divide-And-Conquer Strategy}
\label{sec:algorithm}

In this section, we introduce a brand new solution on large scale $k$-core decomposition, namely \emph{DC-$k$Core}. 

\subsection{Overview}
\begin{figure}[!t]
  \centering
  \includegraphics[width=\linewidth]{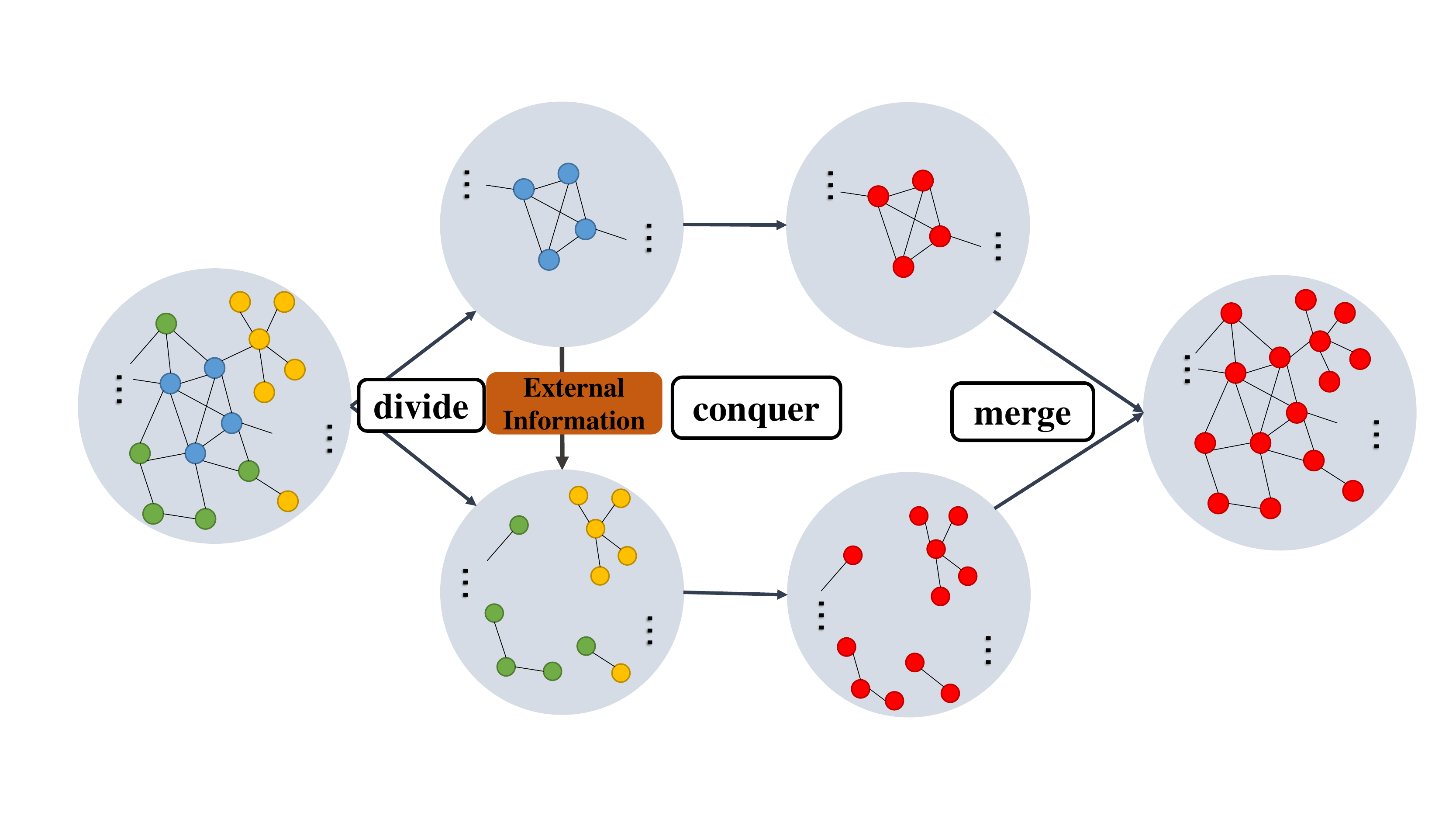}
  \caption{Overview of the divide-and-conquer strategy for $k$-core decomposition.}
  \label{fig:overview}
\end{figure}

As introduced in Section~\ref{sec:intro}, processing a huge graph in one shot is sub-optimal, either due to the resource limitation or the performance degradation. Hence, we propose to split the graph into different parts and process each part individually. 
Figure~\ref{fig:overview} depicts the overview of our work. 
Given a pre-defined threshold $k$, we divide the original graph into two parts, one for the nodes with coreness values not less than $k$ and the other for the remaining nodes. To ensure the correctness, some external information is generated to aid the decomposition, as described later. Finally, the coreness values computed by the two subgraphs are merged together. 

\subsection{The Dividing Step}

\subsubsection{The \textit{Exact-Divide} Strategy}
According to Definition~\ref{def:kcore}, nodes with coreness values smaller than $k$ will never affect the calculation of $k$-core subgraph. Thus, if we extract the $k$-core subgraph $G_k$ out of $G$ and execute decomposition on $G_k$ directly, we can get part of the coreness values of $G$. So, given a threshold $k$, the \textit{Exact-Divide} strategy divides the original graph into two parts --- the $k$-core subgraph $G_k$ and the rest of the graph $G_{!k}$, and apply decomposition to the two subgraphs individually in the conquer step. 

Nevertheless, although the coreness values of nodes in $G_k$ can be obtained directly by decomposing, 
for the rest of the graph, i.e., $G_{!k}$, the removal of nodes with large coreness values will inevitably affect the calculation of coreness values of this part. To correctly compute the coreness values in $G_{!k}$, we propose to make use of the external information generated from $G_k$, as defined in Definition~\ref{def:external_info}.

\begin{figure}[!t]
  \centering
  \includegraphics[width=\linewidth]{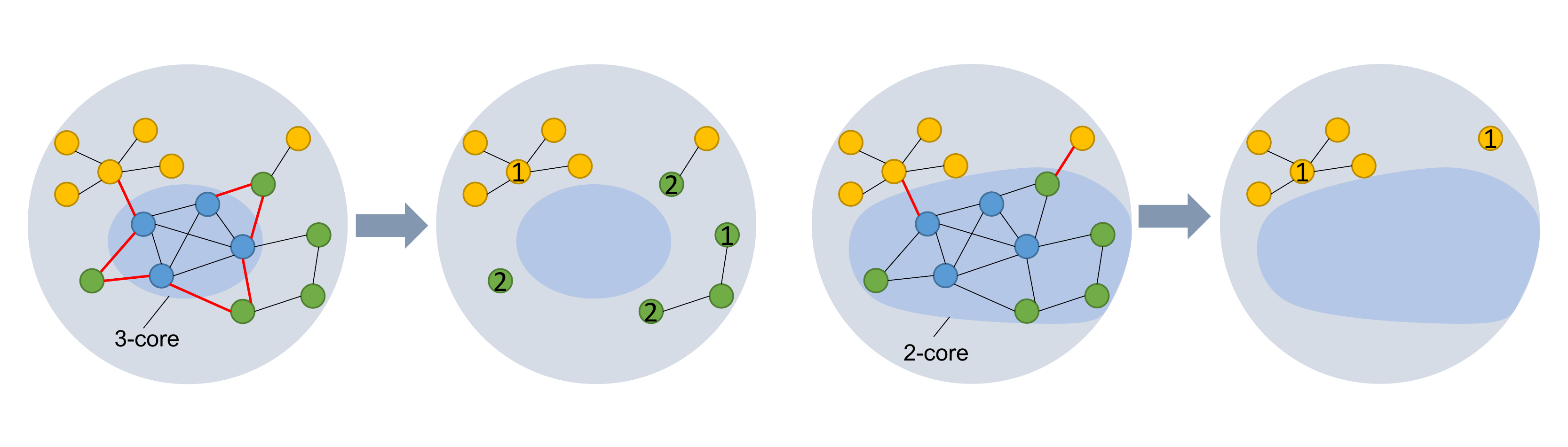}
  \caption{Examples of external information generating for $G_{!k}$ for $k=2$ (left) and $k=3$ (right), respectively. Each node in $G_{!k}$ counts its neighbors in $G_k$.}
  \label{fig:external_info}
\end{figure}

\begin{definition}
\label{def:external_info}
Given $G(V,E)$ and $G_k$, for node $v\in V_{!k}$, the external information of $v$ is defined as $\mathcal{E}(v)=|\{u: u \in \mathcal{N}_G(v) \land u \in V_k\}|$.
\end{definition}

Intuitively, the external information for $v$ summarizes the number of nodes connected with $v$ in $G_k$, which perfectly replaces the subgraph $G_k$ in the calculation of coreness values of nodes in $G_{!k}$ so that we can ensure the correctness of the $k$-core decomposition on $G$. We will show detailed proof of the correctness in Section~\ref{sec:conquerproof}. As illustrated in Figure~\ref{fig:external_info}, the process of external information generation is straightforward --- for each node in $G_{!k}$, we count its neighbors in $G_k$ as the corresponding external information. 

\begin{figure}[!t]
  \centering
  \includegraphics[width=0.75\linewidth]{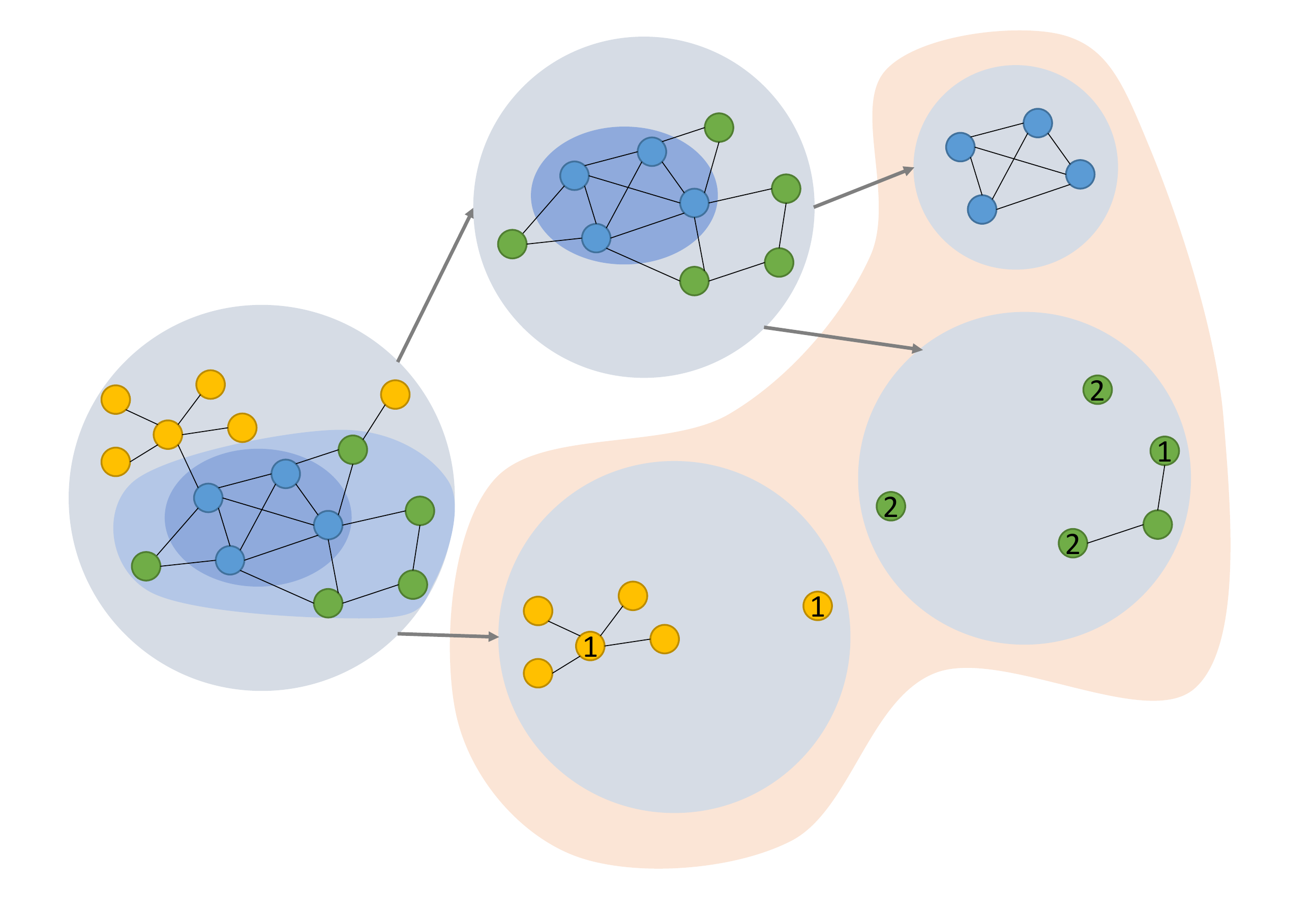}
  \caption{An example of dividing a graph into three parts via the \textit{Exact-Divide} strategy. The numbers represent the external information of the correspoding nodes.}
  \label{fig:divide_sample}
\end{figure}

In summary, by using the \textit{Exact-Divide} strategy, we divide the graph $G$ into two parts, i.e., $G_k$ and $G_{!k}$. For $G_k$, we can directly decompose it, whilst for $G_{!k}$, the external information is generated to aid the decomposition. 
Moreover, if the scale of $G_k$ is still too large, we can continue using the \textit{Exact-Divide} strategy on $G_k$ to divide $G$ into more parts, as illustrated in Figure~\ref{fig:divide_sample}. With the external information, we can calculate the coreness values for each part independently with correctness guarantees on all parts. As a result, huge graphs can be supported within limited resource constraints. 

\subsubsection{The \textit{Rough-Divide} Strategy}
As aforementioned, with the \textit{Exact-Divide} strategy, the graph $G$ is divided into $G_k$ and $G_{!k}$, which represent the $k$-core subgraph and the rest part, respectively. However, it requires to extract the the $k$-core subgraph exactly, which is non-trivial for many large graphs in practice. In order to reduce the time consumption of the dividing process, we propose a heuristic strategy called the \textit{Rough-Divide} strategy. 

In \cite{esfandiari2018parallel} an approximate k-core decomposition is proposed based on sampling. It divides the graph by sampling edges, and provides a good approximation of coreness. This sampling method is complex and unsuitable for our situation. In our \textit{Rough-Divide}, the dividing efficiency is the only concern, thus we provide an easy straight-forward dividing strategy. 

\begin{figure}[!t]
  \centering
  \includegraphics[width=0.8\linewidth]{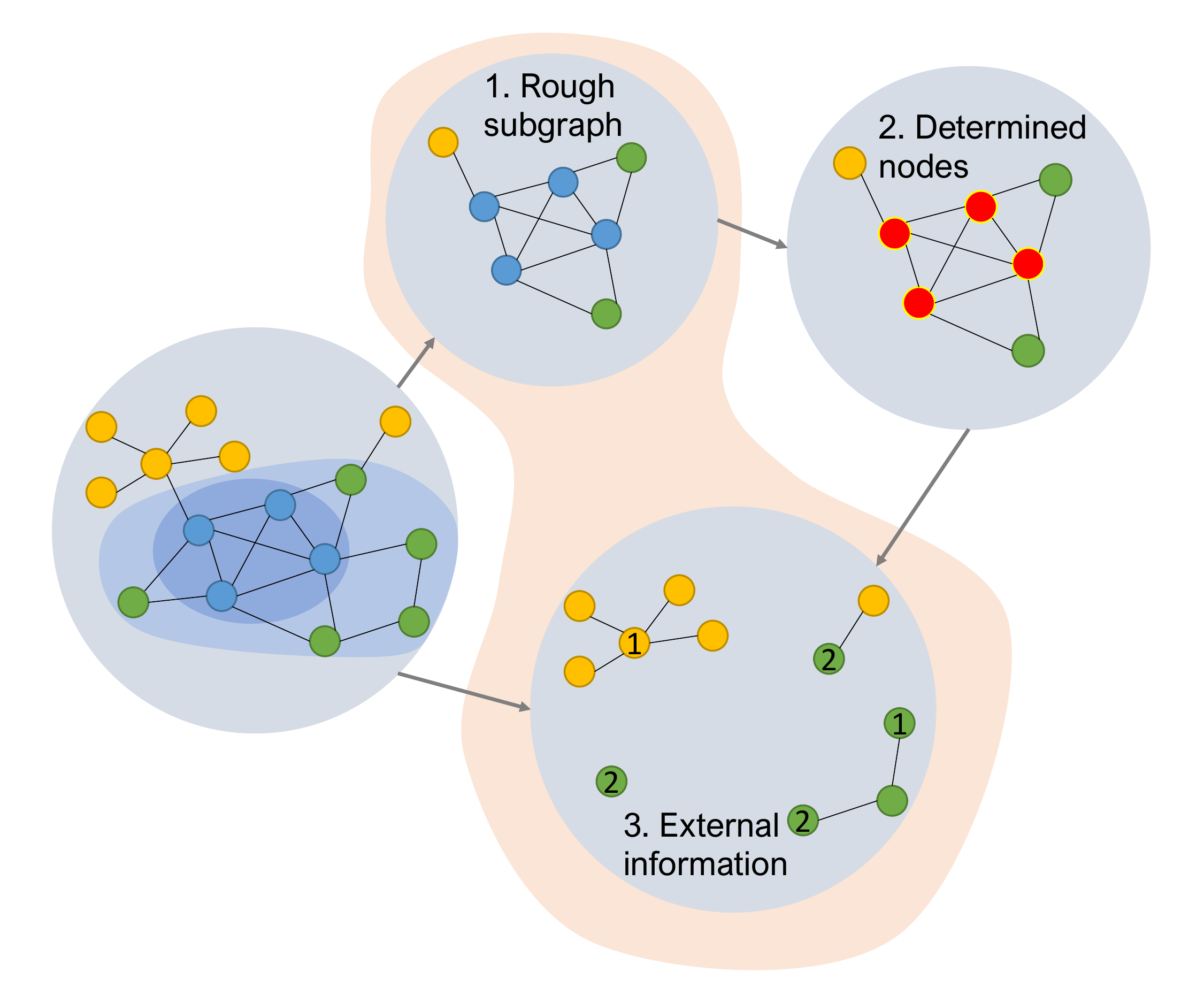}
  \caption{An example of the \textit{Rough-Divide} strategy. A rough subgraph $G_{\hat{k}}$ is extracted and decomposed. Then the nodes in $G_{!k}$ can be determined and the corresponding external information can be generated for the decomposition of $G_{!k}$.}
  \label{fig:seqdivide_sample}
\end{figure}

Unlike the \textit{Exact-Divide} strategy, the \textit{Rough-Divide} strategy does not need to extract the exact $k$-core subgraph $G_k$. Instead, an approximate, rough subgraph $G_{\hat{k}}$ is considered, where $G_{\hat{k}}$ is a superset of the $k$-core subgraph $G_k$, i.e., every edge in $G_k$ belongs to $G_{\hat{k}}$ as well. Obviously, the extraction of $G_{\hat{k}}$ is much more efficient than $G_k$. 
For instance, we can extract the nodes whose degrees are at least $k$, i.e., $G_{\hat{k}} = (C, E|C)$ where $C = \{v|\forall v\in V, deg_G(v) \geq k\}$.
Then the two parts are processed individually, as described below. 

For the first part, i.e., $G_{\hat{k}}$, we observe that since it is the superset of $G_k$, and according to Definition~\ref{def:kcore}, nodes with coreness less than $k$ have zero impact on the calculation of the $k$-core subgraph, decomposing $G_{\hat{k}}$ will yield a superset of the results obtained by decomposing $G_k$.  
As a result, there are two steps to achieve the $k$-core subgraph $G_k$: (i) apply decomposition on $G_{\hat{k}}$, and (ii) remove the nodes with coreness values smaller than $k$ in $G_{\hat{k}}$.
Moreover, the coreness values of all nodes in $G_k$ are achieved simultaneously. 

For the rest of the original graph, we adopt a similar procedure as \textit{Exact-Divide}. In other words, the external information of all nodes in $G_{!k}$ are generated to aid the computation. However, unlike \textit{Exact-Divide}, the external information generation of the \textit{Rough-Divide} strategy requires that we have already accomplished the processing of $G_{\hat{k}}$ so that we can find out all nodes for $G_{!k}$. So, the two parts are processed sequentially.
In practice, such a sequential processing is consistent with many real-world scenarios --- for whom does not have enough resources, different parts are expected to be processed one by one to reduce the peak amount of resources. 

We conclude this section with an example of the \textit{Rough-Divide} strategy in Figure~\ref{fig:seqdivide_sample}. Initially, we extract a rough subgraph $G_{\hat{k}}$ which contains all nodes and edges in $G_k$. Then we decompose $G_{\hat{k}}$ individually. As analyzed above, once the decomposition of $G_{\hat{k}}$ has been accomplished, the coreness values of all nodes that is no less than $k$ can be determined, which make up to be $G_k$. Finally, we can extract the rest part $G_{!k}$ and decompose it in a similar way as the \textit{Exact-Divide} strategy. As we will evaluate in Section~\ref{sec:experiment}, thanks to the acceleration in subgraph extraction, the \textit{Rough-Divide} strategy is more efficient than the \textit{Exact-Divide} strategy.

\subsection{The Conquer Step}
\label{sec:implementation}
Following the dividing step, the conquer step treats each subgraph as a sub-task and decompose it individually. As the decomposition of $G_k$ or $G_{\hat{k}}$ will not be influenced by the nodes in $G_{!k}$, we focus on the decomposition of $G_{!k}$ with the external information, including the correctness analysis and implementation details. 

\subsubsection{Theoretical Analysis}
\label{sec:conquerproof}
Now we explain how the external information could help the computation of coreness values for nodes in $G_{!k}$. By counting the number of neighbors in $G_{k}$, the upper bound of $\mathcal{E}(v)$ is $coreness(v)$, which is summarized in Corollary~\ref{def:coro:external}.

\begin{corollary}
\label{def:coro:external}
Given $G(V, E)$, $G_k(V_k, E_k)$ and $G_{!k}(V_{!k}, E_{!k})$, the upper bound of $\mathcal{E}(v)$ is $coreness(v)$.
\end{corollary}
\begin{proof}
By contradiction. Given a node $v \in V_{!k}$, assume that $\mathcal{E}(v) > coreness(v)$, which means $v$ has at least $coreness(v) + 1$ neighbors belonging to $G_k$, and since $v \not\in G_k$, thus $k > coreness(v)$, that means according to Definition~\ref{def:coreness}, the coreness value of $v$ is at least $coreness(v)+1 \neq coreness(v)$, a contradiction.
\end{proof}

External information for a node takes the same effect as if the missing edges are still connected. For nodes with corenss $k$, only the number of neighbors with coreness larger or equal to $k$ matters.
Based on Corollary~\ref{def:coro:external}, we formalize the correctness of \textit{Exact-Divide} via Corollary~\ref{coro:correct_ness}. 

\begin{corollary}
\label{coro:correct_ness}
Given $G(V,E)$, $G_{!k}(V_{!k}, E_{!k})$ and the external information $\mathcal{E}$, coreness value of $v\in V_{!k}$, $Core_{!k}(v)$, generated by executing a bottom-up algorithm on $G_{!k}$ will converge to the correct coreness $C(v)$ of $v$ in $G$. 
\end{corollary}
\begin{proof}
According to Corollary~\ref{def:coro:external}, we have $\mathcal{E}(v)<k$. And the coreness values of all nodes in $G_{!k}$ are smaller than $k$.
The proof is by induction on the coreness $C(v)$ of node $v \in V_{!k}$.
\begin{itemize}
    \item $C(v)=0$. In this case, $v$ is isolated in both $G$ and $G_{!k}$, and the external information is $\mathcal{E}(v) = 0$. The coreness value $Core_{!k}(v)$ is initialized by $deg_{G_{!k}}(v) + \mathcal{E}(v) = 0$. Thus $C(v) = 0 = Core_{!k}(v)$. The protocol terminates at the very beginning of node $v$.
    
    \item $C(v)=1$. By contradiction, assume that $C(v) = 1$ but $Core_{!k}(v) >= 2$. We have $\mathcal{E}(v) \leq 1$. \\
    If $\mathcal{E}(v) = 1$, node $v$ has only one neighbor belongs to the $k$-core subgraph. However, since $Core_{!k}(v) >= 2$, which means that there exists a neighbor of $v$ in $G$, whose coreness is less than $k$ but greater than 1, which leads to $C(v) > 1$, a contradiction.\\
    If $\mathcal{E}(v) = 0$, node $v$ has no neighbor belongs to the $k$-core subgraph. Since $Core_{!k}(v) >= 2$, which means that there exist two neighbors of $v$ in $G$, whose coreness is less than $k$ but greater than 1, which leads to $C(v) > 1$, a contradiction.
    
    \item \emph{Induction step ($1<C(v)<k$)}. By contradiction, assume that there is a node $v$ satisfying $C(v)=s$ and $Core_{!k}(v)>s$ for all iterations. We have $\mathcal{E}(v) \leq s$ and $Core_{!k}(v) > s$, which means that there exists at least one neighbor of $v$ in $G$, whose coreness is less than $k$ and larger than $s$. Therefore, $v$ have at least $s+1$ neighbors belonging to the $(s+1)$-core subgraph, i.e., $C(v) \geq s+1$, which forms a contradiction.
\end{itemize}
\end{proof}

Now that we can correctly compute the coreness values for all subgraphs, thus the final results can be achieved by merging all the results of subgraphs.

\begin{figure}[!t]
  \centering
  \includegraphics[width=0.75\linewidth]{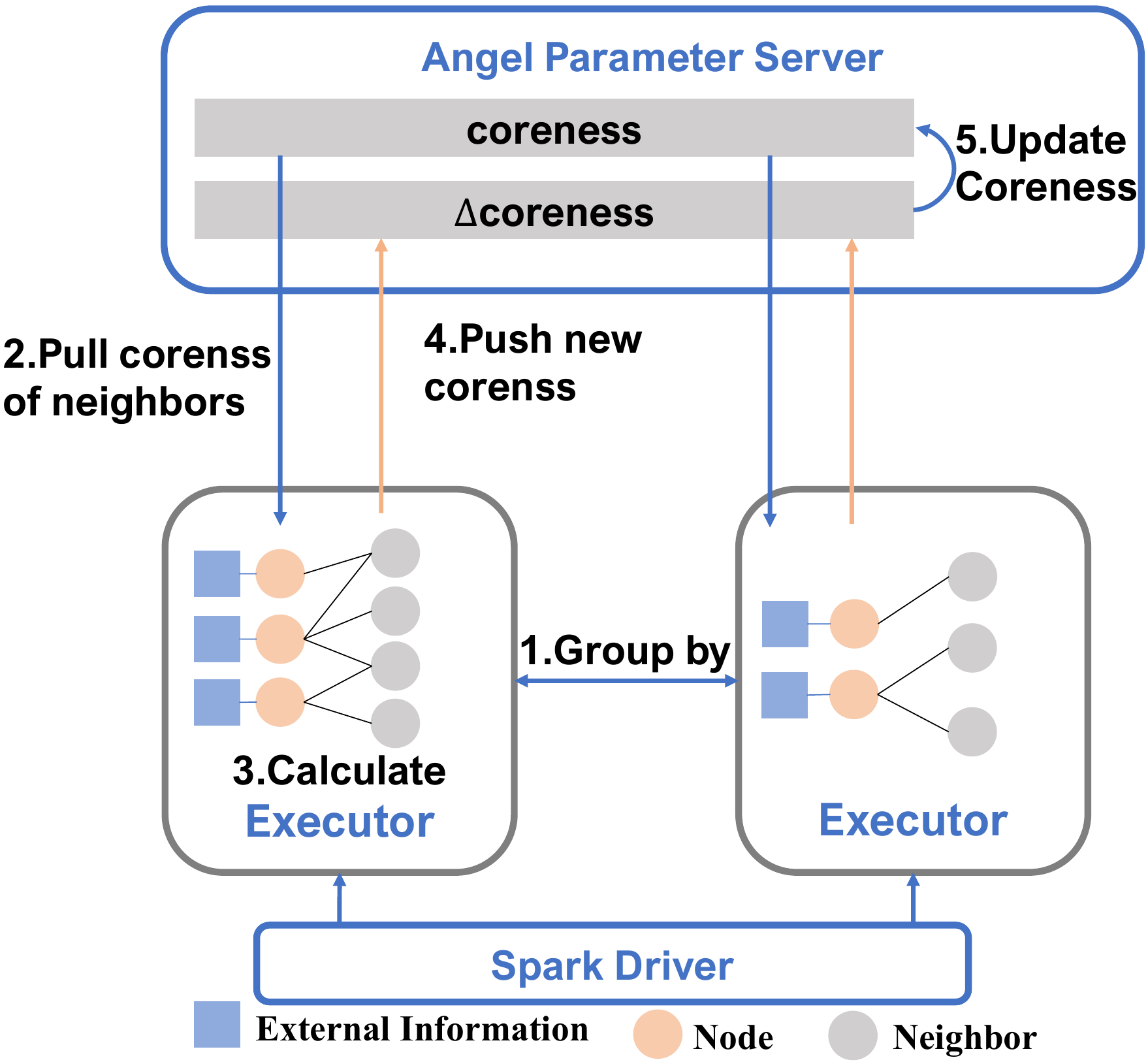}
  \caption{The overview of our implementation of distributed $k$-core decomposition.}
  \label{fig:impl}
\end{figure}

\subsubsection{Implementation}
Since the calculation of coreness values for nodes in $G_{!k}$ requires the external information generated from $G_k$, directly applying decomposition on $G_{!k}$ via Algorithm~\ref{alg:hindexestimate} will yield incorrect results. As a result, we develop Algorithm~\ref{alg:core_external}, which estimates the coreness values with the help of external information. The essential idea is to complement $\mathcal{E}(v)$ to the estimation for the edges connected to $G_k$. Specifically, if node $v$ has no neighbors in $G_k$, we denote $\mathcal{E}(v) = 0$ for simplicity. 

Similar as \emph{PSGraph}, we implement \emph{DC-$k$Core} on top of Apache Spark and the Angel parameter server. As shown in Figure~\ref{fig:impl}, our distributed $k$-core decomposition has the following steps:

\begin{enumerate}
\item \textit{Data Loading}. The graph is loaded and distributed into several partitions in the vertex-centric manner. To be specific, for each partition, we co-located the external information of corresponding nodes in the same RDD partition. 

\item \textit{Pull Coreness}. In each iteration, we sample a batch of nodes and pull the estimated coreness values of their neighbors from the parameter servers. 

\item \textit{Estimate Coreness}. Based on the estimated coreness values of neighbors and the external information, we estimate the coreness values for all nodes in the current batch via Algorithm~\ref{alg:core_external}. 

\item \textit{Push Coreness}. For nodes whose estimated coreness values are different from the previous iteration, we push the updated estimation to parameter servers. 

\item \textit{Update Coreness}. Upon receiving the newly estimated coreness, the parameter servers update the estimation in place for the next iteration.
\end{enumerate}

Step (2)-(5) repeat iteratively until the estimation of all nodes become stable, which finalizes the $k$-core decomposition procedure. 

\begin{algorithm}[!t] 
\caption{Procedure to estimate the coreness value for a node $v$ with external information in one iteration. We denote $\mathcal{E}(v) = 0$ if node $v$ has no neighbors in $G_k$.} 
\label{alg:core_external} 
\begin{algorithmic}[1]
\REQUIRE ~~\\
$\mathcal{N}(v)$: The neighbors of node $v$;\\
$\mathcal{E}(v)$: The external information of node $v$;
\ENSURE ~~\\
Estimation for the coreness value of node $v$;

\STATE $Cores \leftarrow$ Estimation of coreness values of all nodes in $\mathcal{N}(v)$ from the previous iteration;
\STATE $Cores \leftarrow Cores.sorted.reverse$
\STATE $i \leftarrow 0$
\STATE $C_v \leftarrow \mathcal{E}(v) + Cores.length$

\WHILE{$i < Cores.length$}
    \IF{$Cores(i) \geq \mathcal{E}(v) + i + 1$}
        \STATE $i \leftarrow i + 1$;
    \ELSE
        \STATE $C_v \leftarrow \mathcal{E}(v) + i$;
        \STATE break while;
    \ENDIF
\ENDWHILE
\RETURN $C_v$;

\end{algorithmic}
\end{algorithm}


\section{Evaluation in the Real World}
\label{sec:experiment}

As introduced in Section~\ref{sec:intro}, \emph{DC-$k$Core} has been integrated into the production pipeline of our industrial partner and works as a powerful graph analytics toolkit for many real-world applications. In this section, we conduct experiments on both public graph dataset and some super-large-scale use cases to evaluate \emph{DC-$k$Core} empirically. There are algorithms deal with large-scale graphs with secondary storage \cite{cheng2011efficient}\cite{khaouid2015k}, but our method takes memory as the only storage, the comparison with these algorithms is unnecessary.

\subsection{Experimental Setup}
\subsubsection*{Datasets}
As summarized in Table~\ref{tab:datasets}, three graph datasets are used in our experiments, including one public graph dataset and two industrial-scale graph datasets. The \textit{com-friendster}~\cite{yang2015defining} dataset is a public undirected graph dataset with 1.8B edges from Friendster, a social networking site where users can make friends with each other. In this dataset, each node represents an user and each edge represents the friend relationship between the two connected users. \textit{WX-15B} and \textit{WX-136B} are two graphs generated from a payment app of our industrial partner for risk control and advertising. The number of edges of these two industrial-scale datasets are up to 15 billion and 136 billion, respectively, where each edge represents an online payment transaction between two accounts. 

\subsubsection*{Competitors}
Prior to this work, there are two implementations of the distributed $k$-core decomposition algorithm deployed in the production pipeline of our industrial partner. To thoroughly evaluate the effectiveness of our work, we compare all these works, as listed below.
\begin{itemize}[leftmargin=*]
\item \emph{Spark-$k$Core}~\cite{mandal2017distributed}: An implementation of the distributed $k$-core decomposition algorithm on top of the Spark GraphX platform. However, it requires to reconstruct the graph repeatedly during the decomposition task. 
\item \emph{PSGraph}~\cite{jiang2020psgraph}: An implementation of the distributed $k$-core decomposition algorithm on top of Spark and the Angel platform. It leverages the power of the parameter server to update the graph information in place so that graph reconstruction is avoided. 
    
\item \emph{DC-$k$Core}: An implementation of our divide-and-conquer strategy. By default, we divide the original graphs into two parts via the \textit{Rough-Divide} strategy. Furthermore, we adapt \emph{PSGraph} to support decomposition of subgraphs with external information. 
\end{itemize}

\label{sec:datasets}
\begin{table}[!t]
    \caption{Description of datasets.}
    \label{tab:datasets}
    \begin{tabular}{ccccl}
        \toprule
        Dataset & \#Nodes & \#Edges & max k\\
        \midrule
        com-friendster & 65,608,366 & 1,806,067,135 & 304\\
        WX-15B & 646,408,482 & 15,179,911,593 & 401\\
        WX-136B & 2,226,845,928 & 136,588,315,957 & 1,179\\
        \bottomrule
    \end{tabular}
\end{table}

\subsubsection*{Environments}
All experiments are carried out on a productive cluster of 100 nodes. As stated in Section~\ref{sec:intro}, the maximum resources allocated for each task is constrained. To be specific, each node provisions 6 cores and 50GB of RAM for a single task. For \emph{PSGraph} and \emph{DC-$k$Core}, we allocate 4 nodes to work as the parameter servers (ps) by default. We use OpenBox~\cite{DBLP:conf/kdd/LiSZCJLJG0Y0021} to tune the hyper-parameters.

\subsubsection*{Setup for the threshold k}
To better balance the workload on different graph parts, the threshold $k$ is expected to divide graph into parts with similar number of edges. The selection details of $k$ are determined according to the structure information of graphs. Though situation varies with different graphs, for most socal graphs, it is observed that the nodes and edges decrease exponentially with the increase of $k$ in k-core subgraphs. Thus if one cannot obtain the structure information of the graph, we suggest $k$ for the $i$-th part among $P$ parts to be $Deg_{max} \times {(\frac{i}{P})}^2$, where $Deg_{max}$ is the max degree of the graph.

\subsection{Correctness}
First and foremost, we evaluate the correctness of our divide-and-conquer strategy by comparing the results with all competitors. 

On dataset \emph{com-friendster}, we run \emph{Spark-$k$Core}, \emph{PSGraph}, and \emph{DC-$k$Core}, respectively, and record the results of these three implementations. By comparing the coreness values one by one, these results are completely consistent. 
On dataset \emph{WX-15B}, since \emph{Spark-$k$Core} fails to support such a large scale of graphs, we only compare the results of \emph{PSGraph} and \emph{DC-$k$Core}, which are also consistent. 
The experimental results verify our theoretical analysis empirically. As a result, our divide-and-conquer strategy will not harm the correctness of the $k$-core decomposition tasks. 

\subsection{Scalability and Efficiency}
\begin{table}[!t]
    \caption{Comparison of end-to-end running time. We use the symbol ``-'' to indicate that the competitor fails to support the given workload.}
    \label{tab:capacity}
    \begin{tabular}{ccccl}
        \toprule
         & \emph{com-friendster} & \emph{WX-15B} & \emph{WX-136B} \\
        \midrule
        \emph{Spark-$k$Core} & $1.0h$ & - & - \\
        \emph{PSGraph} & $21min$ & $1.9h$ & - \\
        \emph{DC-$k$Core} & $21min$ & $1.4h$ & $27.5h$ \\
        \bottomrule
    \end{tabular}
\end{table}

The scalability is extremely important to large-scale graphs~\cite{fan2020graph,chen2021towards}.
To evaluate the scalability and efficiency of the proposed method, we compare \emph{Spark-$k$Core}, \emph{PSGraph}, and \emph{DC-$k$Core}, on all datasets, and record the total running time in Table~\ref{tab:capacity}. 
For \emph{DC-$k$Core}, we divide the dataset \emph{com-friendster} by the degree threshold of 80, \emph{WX-15B} by 100, and \emph{WX-136B} by 250, respectively. (Note: if a graph is divided by the degree threshold of $t$, it means nodes with coreness values larger than $t$ and nodes with coreness values not larger than $t$ are calculated separately in two subgraphs.)

\subsubsection*{Scalability}
As shown in Table~\ref{tab:capacity}, under the resource constraints, all the three competitors are able to handle the \emph{com-friendster} dataset. However, with the graph scale increasing, \emph{Spark-$k$Core} fails to support the \emph{WX-15B} dataset, 
whilst \emph{PSGraph} fails to process the \emph{WX-136B} dataset.
In contrast, by applying the divide-and-conquer strategy, \emph{DC-$k$Core} can support all the three datasets. 
To the best of our knowledge, this is the first work that supports $k$-core decomposition for super large graphs with more than 100 billion edges. 

\subsubsection*{Efficiency}
In terms of efficiency, both \emph{PSGraph} and \emph{DC-$k$Core} run much faster than \emph{Spark-$k$Core}. As aforementioned, arming with the parameter server architecture, we do not need to reconstruct the graphs constantly during the decomposition, boosting the performance. On the \emph{com-friendster} dataset, the time cost of \emph{PSGraph} and \emph{DC-$k$Core} are comparable, which verifies that our divide-and-conquer also supports small scale graphs well. On the larger dataset \emph{WX-15B}, \emph{DC-$k$Core} takes less time to accomplish the decomposition compared with \emph{PSGraph}. The performance improvement comes from the reduction in graphs. When we divide a graph into subgraphs, each sub-task focuses on a smaller number of nodes, therefore the total number of updates of coreness values becomes smaller. As we will analyze in Section~\ref{sec:interpretability}, it reduces the total communication by applying the divide-and-conquer strategy. Consequently, \emph{DC-$k$Core} is not only more scalable, but also more efficient on processing large-scale graphs than the existing works. 

\subsection{Interpretability}
\label{sec:interpretability}
Extensive experiments have demonstrated that our method is able to improve both the efficiency and scalability of $k$-core decomposition. In order to anatomize its effectiveness, we record various factors during the execution, including the time cost of decomposition and amount of communication, and analyze the results in this section. 
\label{sec:inter_communication}
\begin{figure}[!t]
  \centering
  \includegraphics[width=0.8\linewidth]{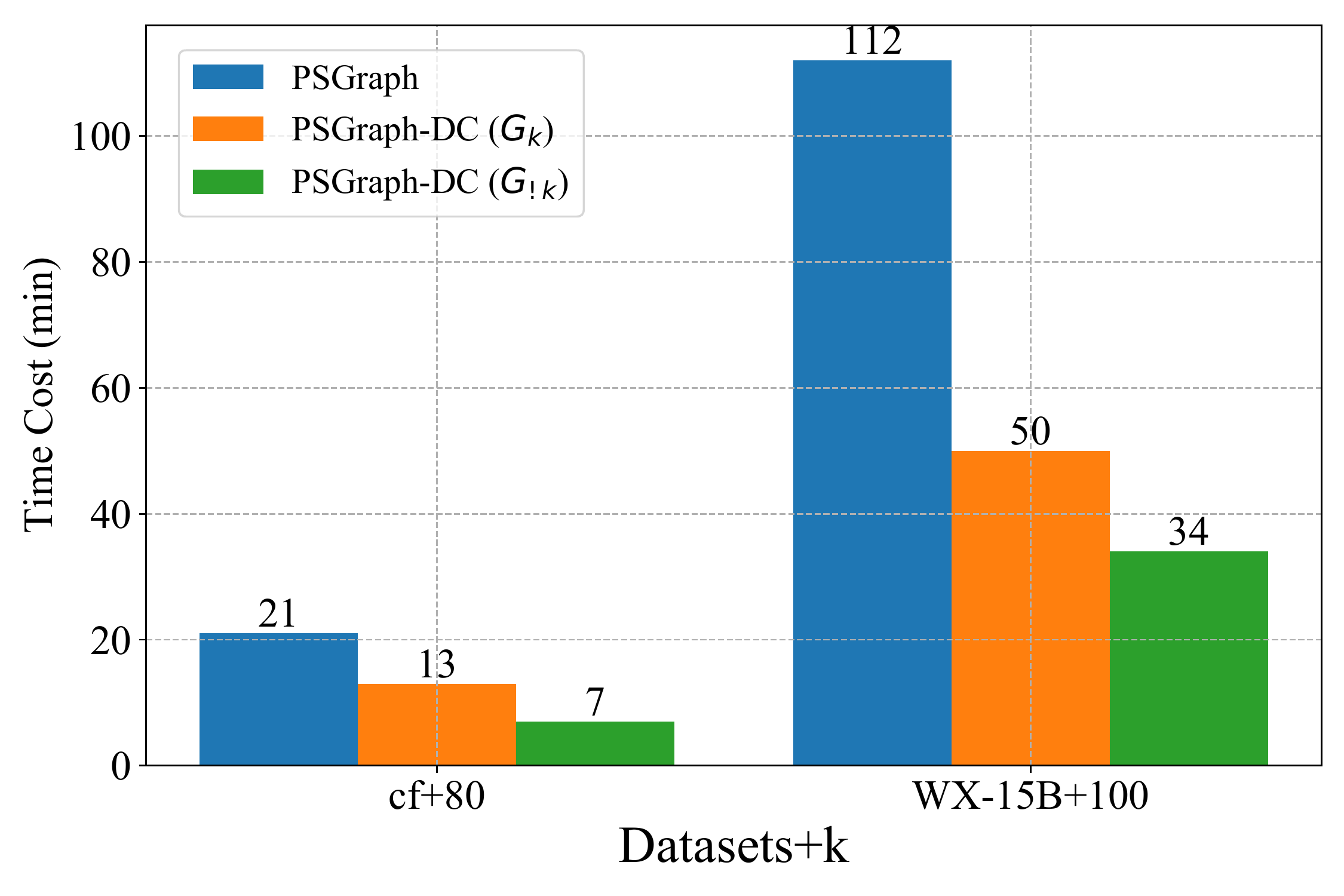}
  \caption{Time cost of decomposition of \emph{PSGraph} and \emph{DC-$k$Core}. For \emph{DC-$k$Core}, we show the time cost of decomposing the two subgraphs $G_{\hat{k}}$ and $G_{!k}$, respectively.}
  \label{fig:iterationtime}
\end{figure}

\begin{figure}[!t]
  \centering
  \includegraphics[width=0.8\linewidth]{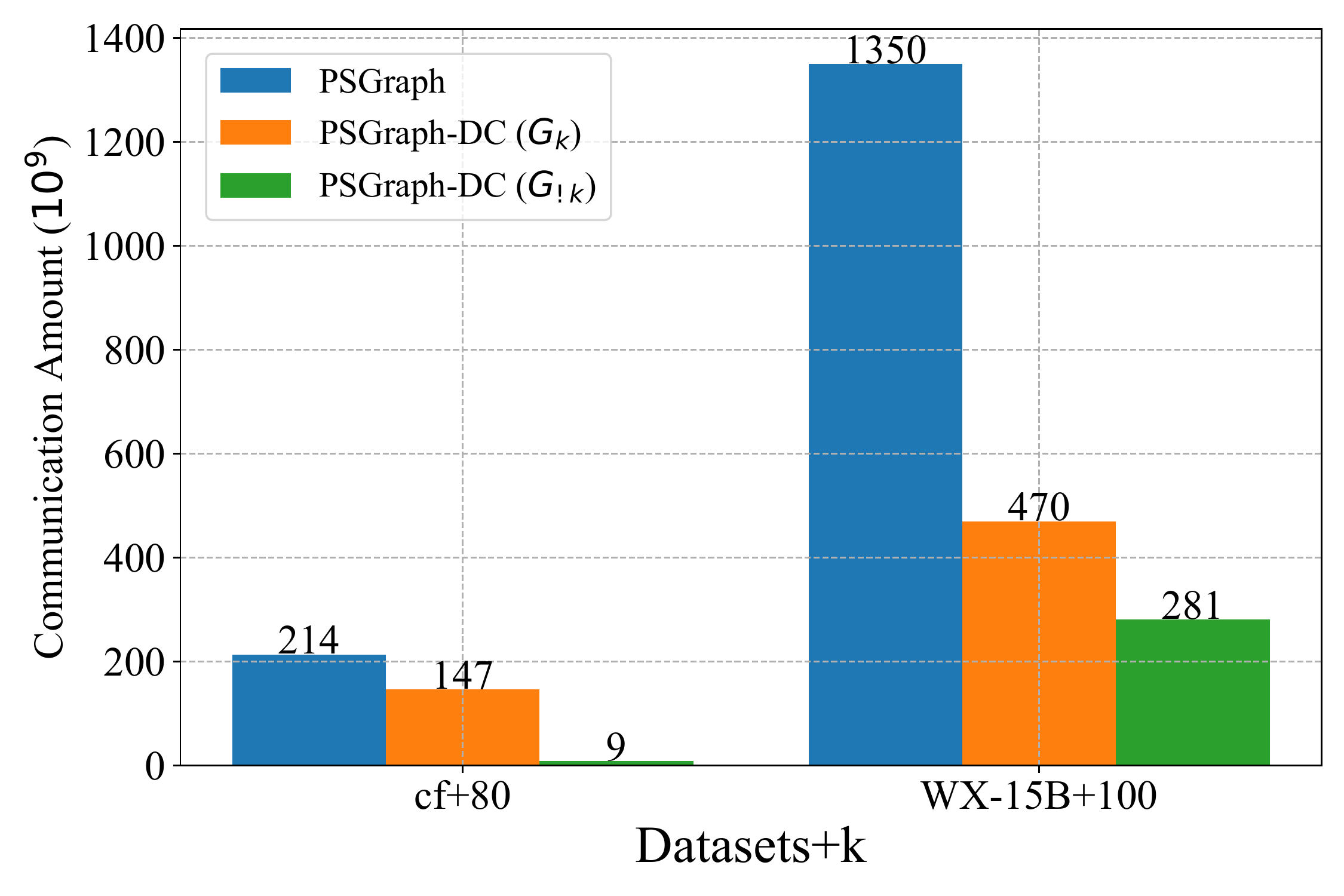}
  \caption{Communication amount of \emph{PSGraph} and \emph{DC-$k$Core}. For \emph{DC-$k$Core}, we show the communication amount of decomposing the two subgraphs $G_{\hat{k}}$ and $G_{!k}$, respectively.
  }
  \label{fig:communications}
\end{figure}

Figure~\ref{fig:iterationtime} summarizes the time cost of performing $k$-core decomposition on the original graph $G$ (i.e., time cost of \emph{PSGraph}) and on the two subgraphs $G_{\hat{k}}$ and $G_{!k}$, respectively. Although the divide-and-conquer requires more than one decomposition task, the total time cost is smaller than decomposing the original graph directly. For instance, for the \emph{WX-15B} dataset, the time cost for $G_{\hat{k}}$ and $G_{!k}$ are 50 minutes and 34 minutes, respectively. The total time cost is 84 minutes, which is 1.3$\times$ faster than decomposing $G$ directly. 

To further analyze the decrease in total time cost, we record the amount of communication of each decomposition task, as shown in Figure~\ref{fig:communications}. To conclude, the total communication amount of \emph{DC-$k$Core} is smaller than \emph{PSGraph} on both datasets. 
In fact, this is not surprising for two reasons. 
First, since when the number of edges is reduced in each subgraph, the number of updated nodes (nodes whose coreness values are updated) in each iteration is also reduced, which eventually results in the reduction in communication. 
Second, for $G_{!k}$, the edges connected to $G_k$ are cut off and the neighborhood is summarized as the external information. During the decomposition, nodes in $G_k$ will not be considered or updated, bringing an extra reduction in communication. This is also consistent with the results in Figure~\ref{fig:communications} --- the communication amount of decomposing $G_{!k}$ is much smaller than $G_{\hat{k}}$. 
Eventually, by breaking the original graph into smaller ones, our divide-and-conquer strategy reduces the communication overhead so that it is more efficient that decomposing the original graph directly.

\subsection{Comparison of Dividing Strategies}

\begin{figure}[!t]
  \centering
  \includegraphics[width=0.8\linewidth]{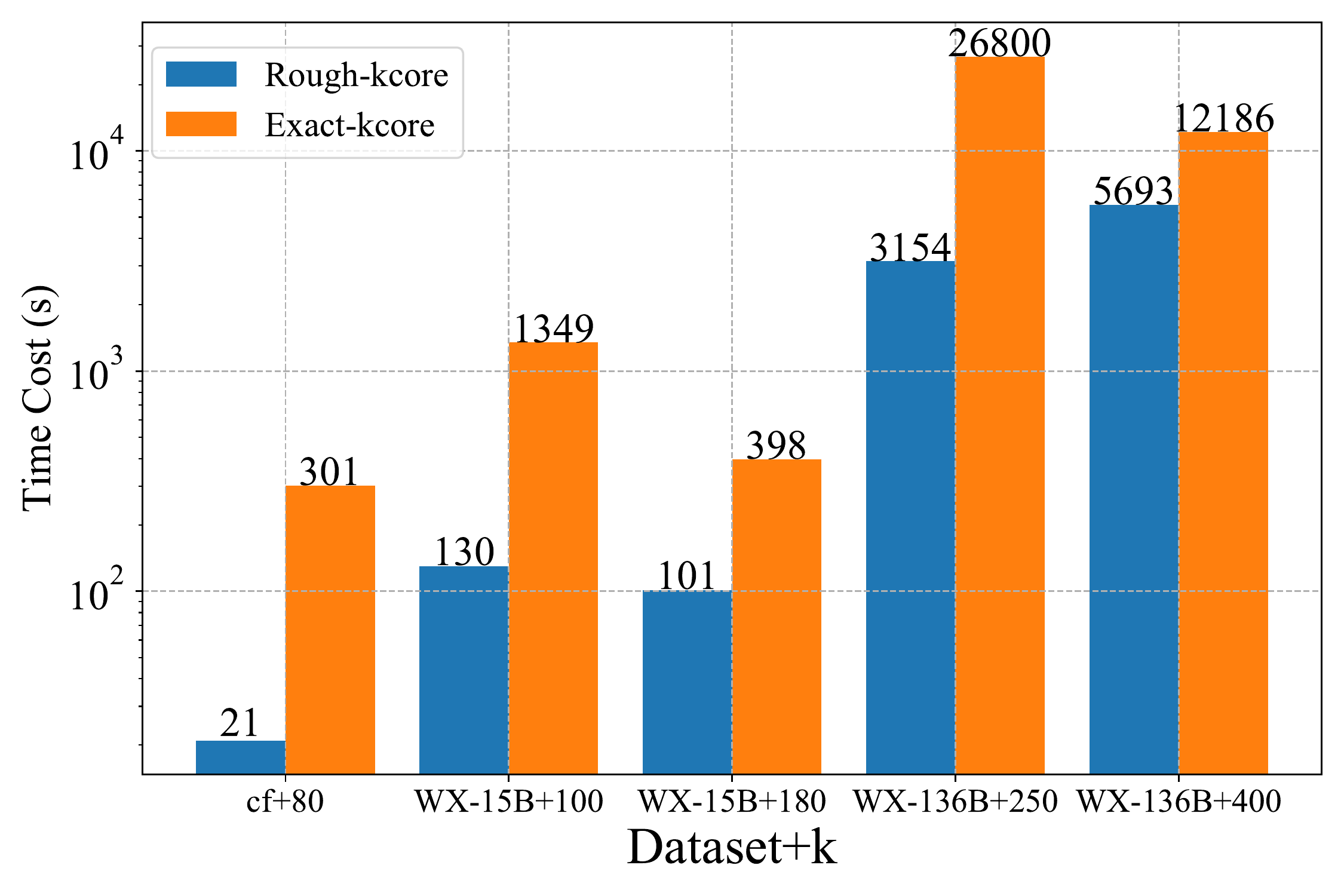}
  \caption{Time cost in seconds (at log scale) for the generation of the rough $k$-core subgraph $G_{\hat{k}}$ and the exact $k$-core subgraph $G_k$. Generating $G_{\hat{k}}$ is $3.7-14.3\times$ faster than $G_k$.}
  \label{fig:roughkcore}
\end{figure}

In Section~\ref{sec:algorithm}, we propose two dividing strategies, i.e., \textit{Exact-Divide} and \textit{Rough-Divide}. 
To compare their efficiency, we record the time cost of subgraph extraction and show the results in Figure~\ref{fig:roughkcore}.

Each group on the X-axis represents the dataset and the degree threshold for subgraph extraction. For example, ``10billion+100'' stands for the generation of rough or exact 100-core subgraph for the \emph{WX-15B} dataset, ``100billion+250'' for the 250-core subgraph of \emph{WX-136B}, and ``cf+80'' for the 80-core subgraph of \emph{com-friendster}, respectively. 
In short, the \textit{Rough-Divide} strategy is $3.7-14.3\times$ faster than the \textit{Exact-Divide} strategy in terms of subgraph extraction. For instance, for ``WX-136B+250'', the rough approach saves 6.5 hours compared with the exact venue. As a result, the \textit{Rough-Divide} strategy is more suitable for huge graphs in practice. 

\subsection{Influence of the Number of Divided Parts}
As discussed in Section~\ref{sec:algorithm}, by applying the divide-and-conquer strategy recursively, a graph can be divided into more than two parts. To evaluate the efficiency improvement of dividing into multiple parts, and to study the overhead of subgraph extraction and external information generation, we conduct experiments to vary the number of parts from 2 to 4 on the \emph{com-friendster} and \emph{WX-15B} datasets. Specifically, for the \emph{com-friendster} dataset, we set the dividing thresholds as 80, 100, and 150, respectively, whilst for the \emph{WX-15B} dataset, the dividing thresholds are set as 80, 100, and 180, respectively.

\subsubsection*{Communication Amount on Multi-parts}
We first analyze the impact on the total amount of communication when the number of dividing parts varies. As shown in Figure~\ref{fig:multiparts}, the total amount of communication of \emph{DC-$k$Core} is always smaller than \emph{PSGraph}, regardless of the number of parts, which shows that the divide-and-conquer strategy can indeed reduce the communication overhead. 

\begin{figure}[!t]
  \centering
  \includegraphics[width=0.8\linewidth]{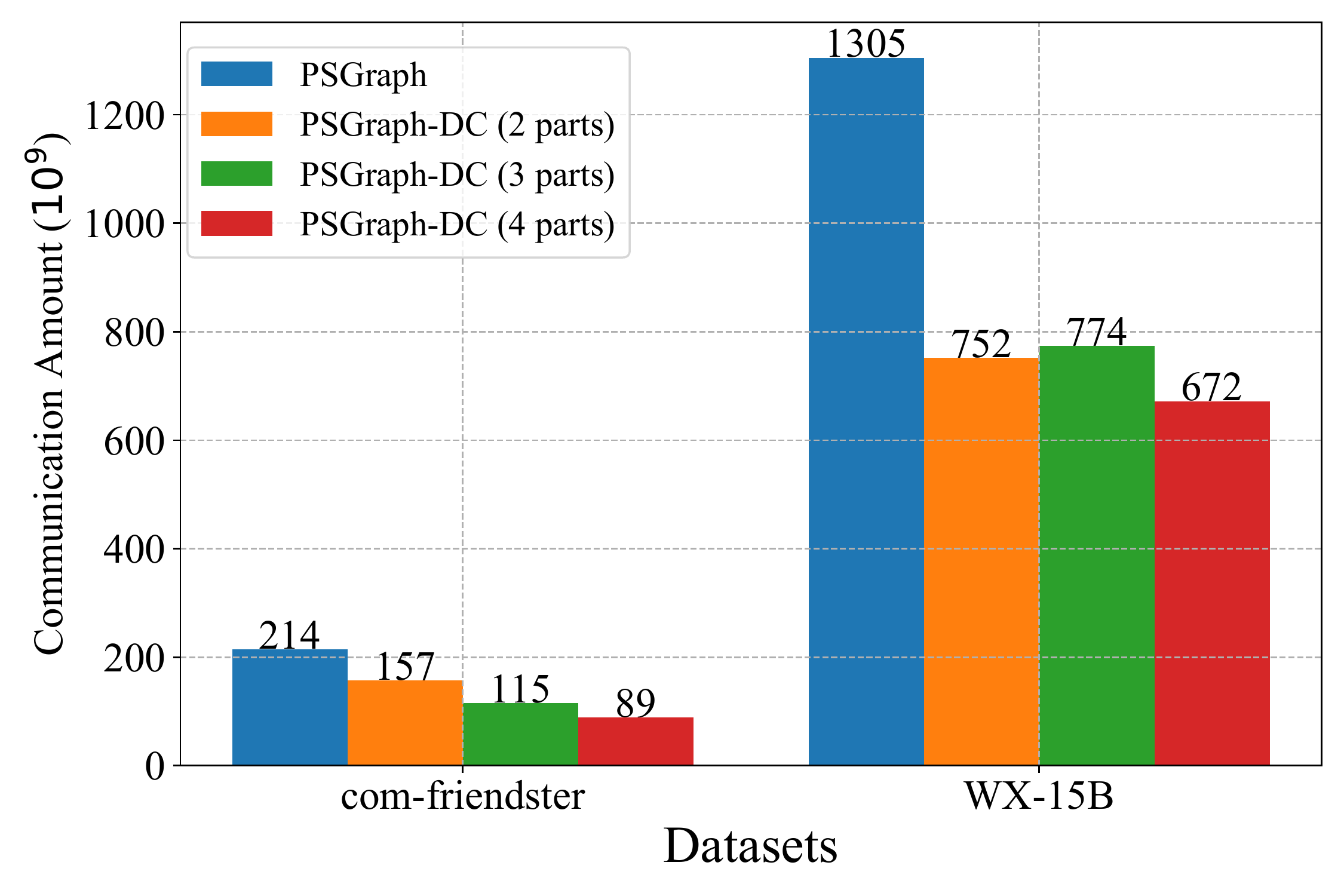}
  \caption{Total amount of communication of \emph{PSGraph} and \emph{DC-$k$Core}. For \emph{DC-$k$Core}, we consider dividing into 2-4 parts (subgraphs) and report the sum of communication amount of all parts.
  }
  \label{fig:multiparts}
\end{figure}

\begin{figure}[!t]
  \centering
  \includegraphics[width=0.8\linewidth]{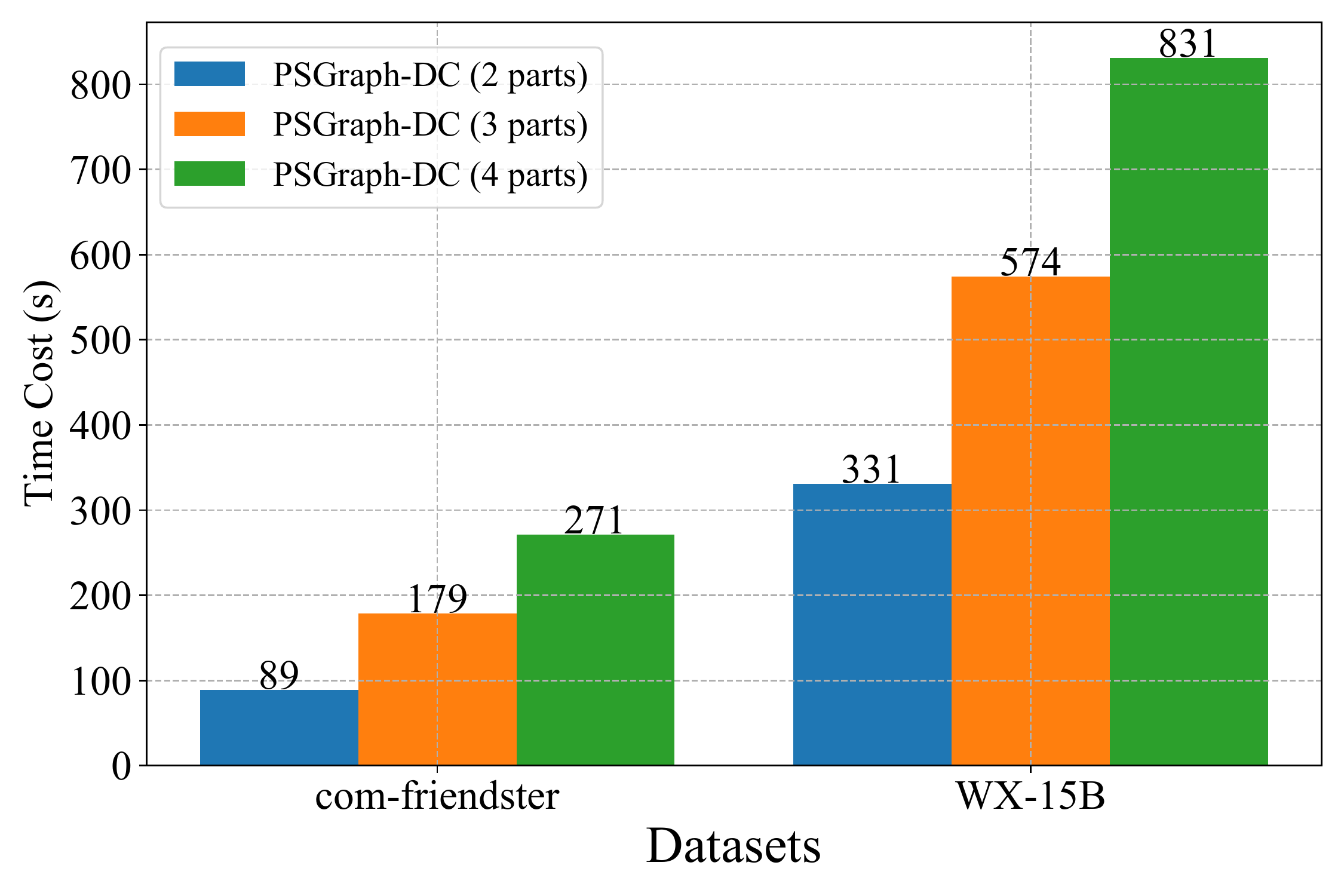}
  \caption{Time cost of preprocessing in minutes, including the time cost of subgraph extraction and external information generation. When the number of parts increases, the preprocessing overhead enlarges.
  }
  \label{fig:mknbrtable}
\end{figure}

For \emph{com-friendster}, the total communication amount decreases when the number of parts increases. For \emph{WX-15B}, although increasing the number of parts from 2 to 3 causes a slight increase in communication, dividing it into 4 parts achieves the smallest communication amount. It is consistent with our analysis in Section~\ref{sec:interpretability} that the reduction in the number of edges brought by graph division can significantly alleviate the communication overhead.

\subsubsection*{Preprocessing Cost}
Next, we record the time cost of preprocessing in Figure~\ref{fig:mknbrtable}, including the extraction of subgraphs and the generation of external information. The results show that when we divide the original graph into more parts, the overhead of preprocessing increases in the meantime. For instance, for the \emph{WX-15B} dataset, the preprocessing cost becomes $2.5\times$ larger when we increase the number of parts from 2 to 4. 
It shows that there exists a trade-off between the communication amount and preprocessing overhead --- a large number of parts will lead to smaller communication amount but heavier preprocessing cost, and vice versa. It would be an interesting topic to investigate how to configure the optimal number of parts for a given workload. We will leave it as our future work.

\section{Conclusion and Future Works}
\label{sec:conclusion}

In this paper, we propose a novel divide-and-conquer strategy for the $k$-core decomposition of huge graphs under limited resource constraints. To reduce the resource requirement, we devise graph division strategies to break the huge graph into smaller subgraphs. To ensure the correctness, we develop an algorithm that decomposes the subgraph with the help of external information. Empirical results show that our work is able to support much larger scale of graphs and is more efficient than the existing works. 

As a possible future direction, we would like to generalize the divide-and-conquer strategy to more graph processing algorithms to develop a robust framework for industrial-scale graph datasets.

\begin{acks}
  This work is supported by NSFC (No. 61832001, 6197200),  Beijing Academy of Artificial Intelligence (BAAI), PKU-Baidu Fund 2019BD006, and PKU-Tencent Joint Research Lab. Bin Cui is the corresponding author.
\end{acks}

\bibliographystyle{ACM-Reference-Format}
\bibliography{sample-bibliography} 

\end{document}